\documentclass[conference]{IEEEtran}
\usepackage{ifpdf}

\usepackage{cite}

\ifCLASSINFOpdf
  \usepackage[pdftex]{graphicx}
\else
\fi
\usepackage{amsmath}
\usepackage{amsfonts}
\usepackage{amssymb}
\usepackage{amsthm}
\usepackage{color}

\usepackage{algorithmic}
\usepackage{array}
\usepackage{url}
\newtheorem{theorem}{Theorem}

\newtheorem{remark}{Remark}

\newcommand{\cV}{\mathcal{V}}
\newcommand{\cE}{\mathcal{E}}
\newcommand{\cB}{\mathcal{B}}
\newcommand{\cR}{\mathcal{R}}
\newcommand{\bL}{\mathbf{L}}
\newcommand{\bP}{\mathbf{P}}

\begin{document}
%
\title{Communication Cost of Transforming a Nearest Plane Partition to the Voronoi Partition}
\author{\IEEEauthorblockN{Vinay A. Vaishampayan}
\IEEEauthorblockA{Engineering Science and Physics Dept.\\City University of New York-College of Staten Island\\Staten Island, NY USA}
\and
\IEEEauthorblockN{Maiara F. Bollauf }
\IEEEauthorblockA{Institute of Mathematics, Statistics and Computer Science\\
University of Campinas, Sao Paulo, Brazil}
}
\maketitle

\begin{abstract}
We consider the problem of distributed computation of the nearest lattice point for a two dimensional lattice.
An interactive model of communication is considered. We address the problem of reconfiguring a specific rectangular partition, a nearest plane, or Babai, partition, into the Voronoi partition.   Expressions are derived for the error probability as a function of the total number of communicated bits. With an infinite number of allowed communication rounds, the average cost of achieving zero error probability is shown to be finite. For the interactive model, with a single round of communication, expressions are obtained for the error probability as a function of the bits exchanged. We observe that the error exponent  depends on the lattice.
\end{abstract}



%
\IEEEpeerreviewmaketitle
\section{Introduction}

Given a lattice~\footnote{A lattice is a discrete additive subgroup of $\mathbb{R}^n$. The reader is referred to \cite{SPLAG} for details.} $\Lambda \subset \mathbb{R}^n$, the closest lattice point problem is to find for each $x=(x_1,x_2,\ldots,x_n)\in \mathbb{R}^n$, the point $\lambda^*(x)$ which minimizes
the Euclidean distance $\|x-\lambda\|$, $\lambda \in \Lambda$. Here, we assume that $x_i$ is available at node $S_i$ in a network of nodes and study the communication cost of this search.  We consider an interactive model in which each node $S_i$ communicates with every other node so that every node  such that each node can determine $\lambda^*(x)$.  Since this may not be possible in general, let $\lambda(x)$ denote the lattice point determined by concerned nodes when computation is halted. The objective is to determine the tradeoff between the  communication required  and the probability of error $P_e:=Pr(\lambda(X)\neq \lambda^*(X))$ for a known probability distribution on $X$.

 We will assume that generator matrix $V$ of $\Lambda$ has the upper triangular form $$V=\begin{pmatrix} 1 & \rho \cos \theta \\ 0 & \rho \sin \theta\end{pmatrix}$$ where the columns of $V$ are basis vectors for the lattice.   The associated quadratic form is $f(x,y)=x^2 + 2\rho \cos \theta~ xy + \rho^2 y^2$. It is known that this form is reduced if and only if $2|\rho \cos \theta|  \leq 1 \leq \rho^2 $ and the three smallest values taken by $f$ over integer $u=(x,y) \neq 0$ are $1$, $\rho^2$, and $1-2|\rho \cos \theta| + \rho^2$~see e.g. Th. II, Ch. II, \cite{Cassels:1997}. Based on a result due to Voronoi, Th. 10, Ch. 21,~\cite{SPLAG}, it follows that the relevant vectors, i.e. the vectors which determine the faces of the Voronoi cell, are $\pm (1,0)$, $\pm (\rho \cos \theta, \rho \sin \theta)$ and $\pm (\rho \cos \theta -1,\rho \sin \theta)$. We thus consider lattices with generator matrix $V$ as above, with $\rho \geq 1$. From an additional symmetry, and in order to avoid indeterminate solutions we restrict $\theta$  such that $0 < \rho \cos \theta < 1/2$. Performance at the endpoints $0$ and $1/2$ can be obtained by taking limits. More generally, the generator matrix of the lattice is represented by matrix $V$ with $i$th column  $v_i$, $i=1,2,\ldots,n$. Thus $\Lambda=\{Vu,~u\in \mathbb{Z}^n\}$.  The $(i,j)$ entry of V is $v_{i,j}$, thus $v_i=(v_{1i},v_{2i},\ldots,v_{ni})$. The Voronoi cell $\cV(\lambda)$ is defined as the set of all $x$ for which $\lambda\in \Lambda$ is the closest lattice point.

In a companion paper~\cite{Bollauf:2017} we have developed upper bounds for the communication complexity of constructing a specific rectangular partition for a given lattice along with a closed form expression for the error probability $P_e$. The partition is referred to as a Babai partition and is an approximation to the Voronoi partition for a given lattice. 

%
%
%
%
%

The remainder of the paper is organized as follows. Previous work is presented in Sec.~\ref{sec:previous}, assumptions and a preliminary analysis are presented in Sec.~\ref{sec:largescale}, the interactive model is analyzed and quantizer design is presented  for a single round of communication (Sec.~\ref{sec:onedec}), for unbounded rounds of communication (Sec.~\ref{sec:infinitedec}).  Numerical results and a discussion are in Sec.~\ref{sec:discussion}. A summary and conclusions is provided in Sec.~\ref{sec:summary}.

\section{Previous Work}
\label{sec:previous}
Communication complexity~\cite{Yao:1979},~\cite{KushNis:1997} is the minimum amount of communication required to compute a function  in a distributed setting. Information theoretic characterizations of  communication complexity  are developed for the two node problem   in~\cite{Orlitsky:2001}. Two models are considered: a centralized model, and an interactive model where two messages are exchanged (one round of communication in our model).  Two terminal interactive communication is studied in considerable detail in \cite{MaIshwar:2011}, and the benefit of an unbounded number of messages is demonstrated. An important and relevant contribution  in ~\cite{MaIshwar:2009},~\cite{MaIshwar:2011} is the  the strict benefit that interactive communication provides for the computation of the Boolean AND function. Another stream of related work has origins in asymptotic quantization theory. The problem of fine quantization for detection problems is addressed in~\cite{Poor:1988},~\cite{Benitz:1989} and ~\cite{Gupta:2003}. More recently, the design of fine scalar quantizers for distributed function computation with a squared error distortion measure is considered in~\cite{Misra:2011} and succeeding works. Significant benefits, especially in the interactive setting are obtained.

\section{Preliminary Analysis}
\label{sec:largescale}
\begin{figure}[h!]
\begin{center}
		\includegraphics[height=3.8cm]{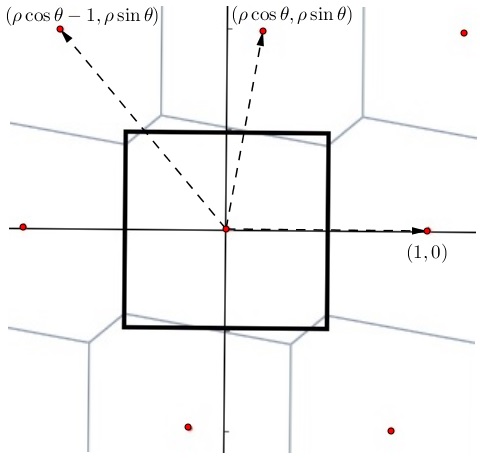}  
\caption{{Voronoi region, Babai partition and three relevant vectors}}
\label{rvectors}
\end{center}
\end{figure}

We consider a two-stage approach for determining $\lambda(x)$.  In Stage-I,~\cite{Bollauf:2017}, a point $\lambda_{np}(x)$ (defined next) is determined, using the nearest plane algorithm~\cite{babai}, and assuming an interactive model. We refer to  $\lambda_{np}(x)$ as the Babai point. When the generator matrix for $\Lambda$ is in upper triangular form,  the nearest plane algorithm determines $\lambda_{np}=Vu$, with $u_i=[(x_i-\sum_{j=i+1}^n v_{ij}u_j)/v_{ii}]$, $i=n,n-1,\ldots,1$ ($[x]$ is the nearest integer to $x$). The nearest plane algorithm partitions the plane into congruent \textit{Babai} cells (rectangles), each of volume $|\det V|$. The Babai cell associated with lattice vector $\lambda$ is denoted $\cB(\lambda)$. Once again we denote the nearest lattice point by $\lambda^*(x)$.

The analysis of $P_{e,I}$, the error at the conclusion of Stage-I, justifies the modeling assumptions made for the Stage-II analysis. Specifically, 
$P_{e,I}=\sum_{\lambda \in \Lambda}Pr(\lambda^*(X) \neq \lambda |X \in \cB(\lambda))Pr(X \in \cB(\lambda))$. 
Since an error occurs if $X$ is closer to some $\lambda' \neq \lambda$, it follows that $P_{e,I}=\sum_{\lambda \in \Lambda}\sum_{\lambda' \neq \lambda}Pr(X \in \cB(\lambda) \bigcap \cV(\lambda')|X \in \cB(\lambda))Pr(X \in \cB(\lambda))$. Assuming that $p(x)$ is approximately constant over each $\cB(\lambda)$\footnote{This is justified under the assumption that the lattice point density is suitably high.}, it follows that $P_{e,I} \approx  \sum_{\lambda \in \Lambda}\sum_{\lambda' \neq \lambda} \mbox{Area}(\cV(\lambda')\bigcap \cB(\lambda)/\mbox{Area}(\cB(\lambda))Pr(X \in \cB(\lambda))$. Since the Babai and Voronoi partitions are invariant under translations by lattice vectors, it follows that $P_{e,I} \approx  \sum_{\lambda' \neq 0} \mbox{Area}(\cV(\lambda')\bigcap \cB(0))/\mbox{Area}(\cB(0))$.

From the above analysis for $P_{e,I}$, and assuming an \textit{interactive model} for Stage-I, it follows that at the end of Stage-I, $\lambda_{np}(X)$ is known to each node $S_i$. Each node thus subtracts off $i$th coordinate $\lambda_{np,i}$ from $X_i$. The result is also referred to as $X_i$ for notational convenience. We will assume that the resulting $X=(X_1,X_2)$ is uniformly distributed over  $\cB(0)$. For the lattice that we consider $\cB(0)=(-1/2,1/2]\times (-(\rho/2) \sin \theta, (\rho/2) \sin \theta]$. Since  $\cB(0)$ has length $L=1$ and height $H=\rho \sin \theta$, we have $p(x_1)=p=1/L$ and $q(x_2)=q=1/H$, where $p,q$ are the marginal pdf's of $X_1$ and $X_2$, respectively. Note that since $\cB(\lambda)$ is rectangular, $X_1$ and $X_2$ are independent. 

Stage-II communication is broken up into \textit{rounds}, one round corresponds to two messages, one from each node in a predefined order. Both orderings, $12$ and $21$  will be considered.

\section{Interactive, Single Round of Communication}
\label{sec:onedec}
At the conclusion of Stage-I, node $S_i$ is in possession of $X_i$, $i=1,2$, and  $X\in \cB(0)$.  We denote the rectangular cells of the partition at the conclusion of Stage-II by $\cR(i),~i=1,2,\ldots,R$. Associated with each cell $\cR(i)$ is a decision $\lambda(i)$. Following steps similar to the analysis above, it follows that $P_{e,II}$, the error probability at the conclusion of Stage-II is given by 
\begin{equation}
P_{e,II}=\sum_{i=1}^R \sum_{\lambda' \neq \lambda(i)}\mbox{Area}(\cR(i)\bigcap \cV(\lambda'))/\mbox{Area}(\cB(0)).
\end{equation}
The optimum decision rule follows immediately: $\lambda(i)=\arg \max_\lambda \mbox{Area}(\cR(i)\bigcap \cV(\lambda))$.

\begin{figure}[htbp] 
   \centering
   \includegraphics[width=2in]{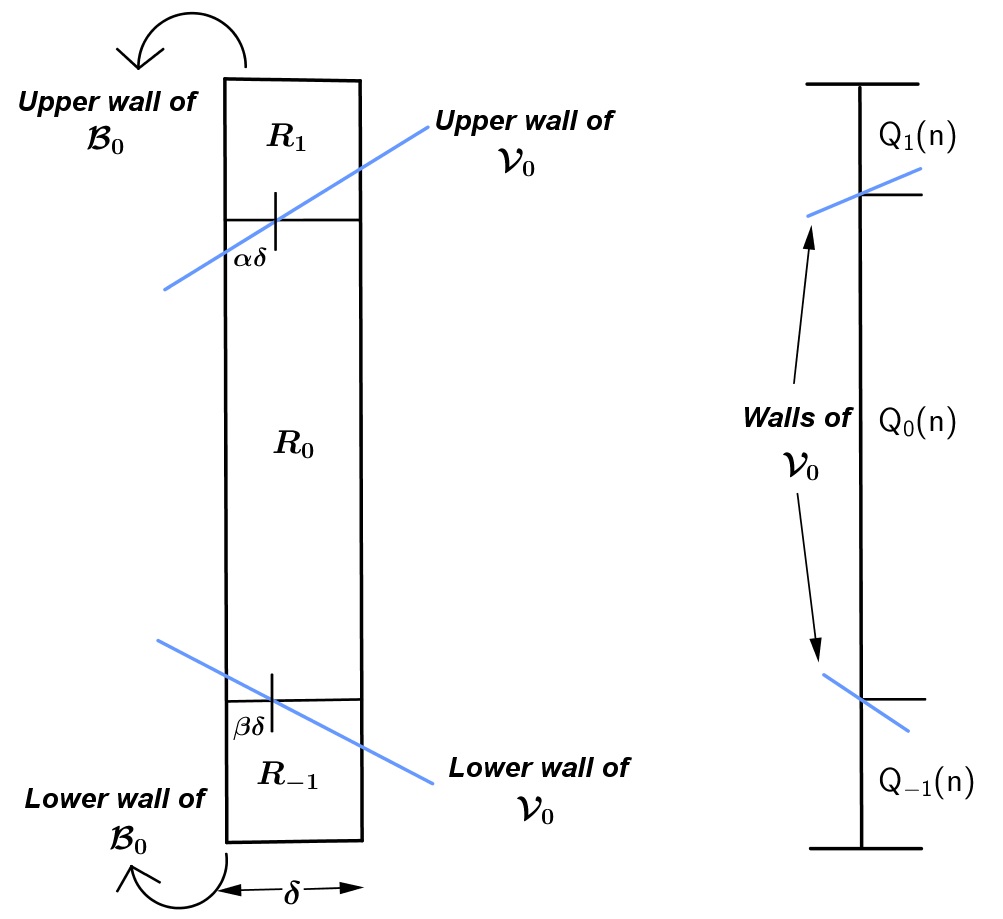} 
   \caption{A  typical vertical strip created by $S_1$ and its partition into three parts by $S_2$ (left). Probability distribution $Q(x)$ which underlies the calculation of $H(U_2|U_1)$ is on the right.}
   \label{fig:verticalslice}
\end{figure}

The scheme for the $12$ order is described first.  To begin, node $S_1$ sends $U_1=i$  to $S_2$ indicating an interval of length $\delta_i$ that $X_1$ lies in. This effectively partitions $(-1/2,1/2]$, the support of $X_1$ into cells of length $\delta_i$ (and equivalently, partitions $\cB(0)$ into vertical strips of widths $\delta_i$), $i=1,2,\ldots,N$. Based on this information, $S_2$ makes a decision $\lambda(U_1,X_2)$ and communicates this decision back to $S_1$ using message $U_2$.  Effectively, $S_2$ partitions each aforementioned vertical strip into at most three parts using at most two horizontal cuts or thresholds. The location of each cut is determined by the location of the appropriate boundary wall of $\cV(0)$. A typical situation is illustrated in Fig.~\ref{fig:verticalslice}. In this figure,  vertical strip $i$ is partitioned into three rectangles, $R_{0}(i)$, $R_{-1}(i)$ and $R_{1}(i)$, where $R_0(i)$ corresponds to points decoded to $\lambda=0$, the other two rectangles are decoded to neighboring points. The probability of error event $\cE$ is given by
\begin{equation}
P_{e,II}=\sum_i \sum_j Pr(\cE|X\in R_{j}(i)),
\end{equation}
where $i$ indexes the strips and $j$ indexes the rectangles within a strip. For the cuts shown in Fig.~\ref{fig:verticalslice}, and assuming the boundary lines have slopes $s_1$ and $s_2$ and bin size $\delta$ we get
\begin{eqnarray}
Pr(\cE|X\in R_{j}(i)) &  =  & \frac{\delta^2}{2 |\det V|}[(\alpha^2 +(1-\alpha)^2)|s_1|+ \nonumber \\
& & ~~~~(\beta^2+(1-\beta)^2)|s_2|] \nonumber \\
& \geq & \delta^2(|s_1|+|s_2|)/4 |\det V|,
\end{eqnarray}
and equality holds when $\alpha=\beta=1/2$.
%
Thus
\begin{equation}
P_{e,II}=(1/4)\sum_{l=1}^2 \sum_{i=1}^N |s_{l,i}|\delta_i^2.
\label{eqn:prob}
\end{equation}

The information rate for Stage-II communication is then $R=H(U_1)+H(U_2|U_1)$.  
The information rates are calculated next. Node $S_1$ sends $H(U_1)$ bits where
\begin{equation}
H(U_1)=\sum_{i=1}^N (\delta_i/L) \log_2 (L/\delta_i).
\label{eqn:rate1}
\end{equation}
Node $S_2$ sends $H(U_2|U_1)$ bits which is obtained by averaging the entropy $H(Q(x_1))$ of probability distribution $Q(x)=(Q_{-1}(x),Q_0(x), Q_1(x))$ over bins of $X_1$ by
\begin{equation}
H(U_2|U_1)= \sum_{i=1}^N (\delta_i/L) H(Q(x_{i})),
\label{eqn:rate2}
\end{equation}
where $x_{i}$ is, say, the midpoint of  the $i$th bin.
Here $Q_1(x), Q_{-1}(x), Q_0(x)$ are the probabilities that $X_2$ exceeds the upper threshold, is smaller than the lower threshold and lies in between the two thresholds, respectively, given $X_1=x$.

\begin{figure}[htbp] 
   \centering
   \includegraphics[width=3.1in]{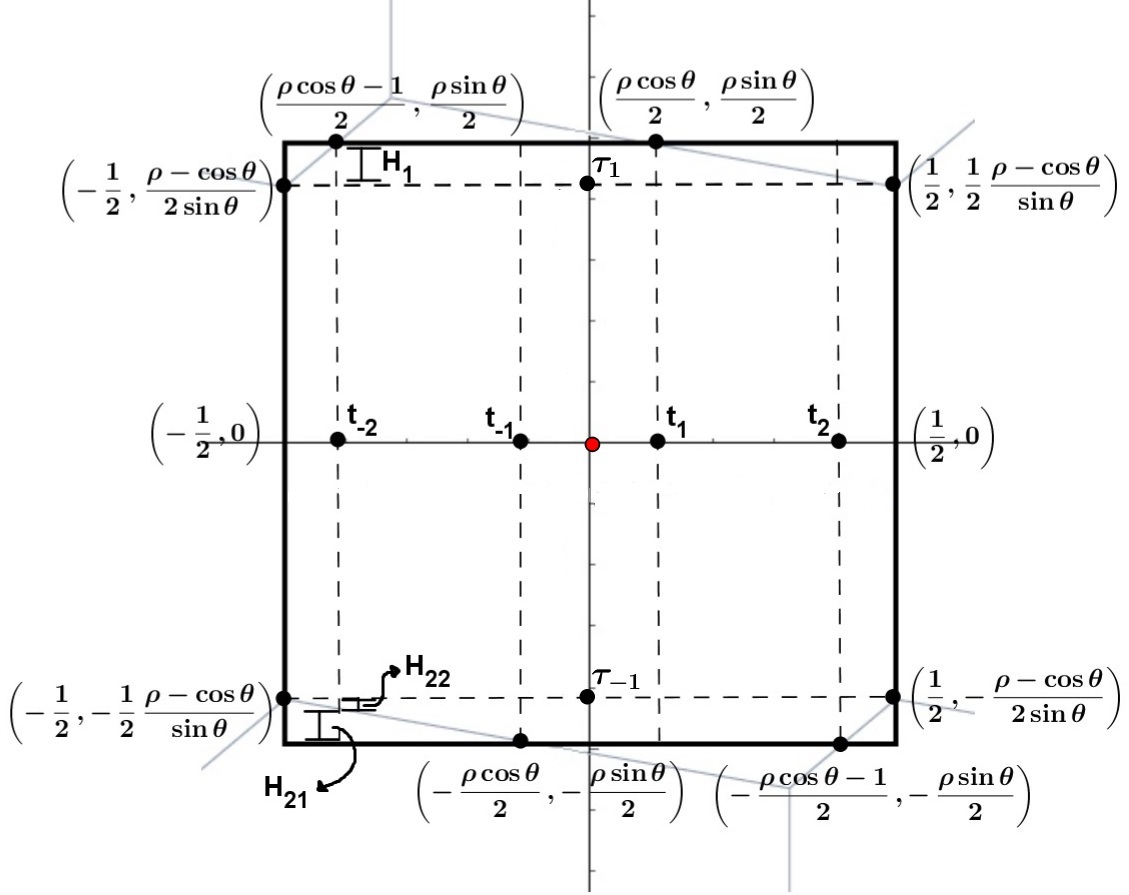} 
   \caption{Babai and Voronoi cells, with key points labeled. $x_1,x_2$ are the horizontal, vertical  coordinates, resp.}
   \label{fig:detailedGeom}
\end{figure}
We now specialize the analysis to $\cV(0)$ and $\cB(0)$ for the given lattice. The geometry of the lattice, with all the significant boundary points, lengths, heights, and slopes  is shown in Fig.~\ref{fig:detailedGeom}. We identify four thresholds $t_{-2}=(\rho \cos \theta-1)/2$, $t_{-1}=(-\rho \cos \theta)/2$, $t_1=-t_{-1}$ and $t_2=-t_{-2}$ and five intervals  $I_{-2}=(-1/2, t_{-2}]$, $I_{-1}=(t_{-2},t_{-1}]$, $I_0=(t_{-1},t_1]$, $I_1=(t_1,t_2]$ and $I_2=(t_2,1/2]$. We partition $I_{-2}$ into $N_2$ equal-length intervals, $I_{-1}$ into $N_1$ equal-length intervals, $I_0$ into $1$ interval, $I_1$ into $N_1$ equal-length intervals and  $I_2$ into $N_2$ equal-length intervals. The lengths of the intervals $I_0$, $I_1$ and $I_2$ are denoted $L_0$, $L_1$ and $L_2$, respectively and $L=L_0+2L_1+2L_2$. Let $\bL=(L_0,L_1,L_1,L_2,L_2)/L$. Note that $\bL$ behaves  like a probability distribution.
Also from Fig.~\ref{fig:detailedGeom}, $H_1=\cos \theta (1-\rho \cos \theta)/2 \sin \theta$, $H_{22}=\rho \cos^2 \theta/2 \sin \theta$ and $H_{21}=\cos \theta (1-2 \rho \cos \theta)/2 \sin \theta$.

From (\ref{eqn:prob}), it follows that
\begin{equation}
P_{e,II}=\alpha_1/N_1+\alpha_2/N_2,
\end{equation}
where $\alpha_1=L_1(H_1+H_{22})/2|\det V|$, $\alpha_2=H_{21}L_2/2|\det V|$.
From (\ref{eqn:rate1}) it follows that 
\begin{equation}
H(U_1)=H(\bL)+(2L_1/L)\log_2 N_1+(2L_2/L) \log_2 N_2.
\end{equation}
In order to calculate $H(U_2|U_1)$, we write $U_1=(V,W)$, where $V$ identifies the interval $I_{V}$, (one of $I_{-2},I_{-1},I_0,I_1,I_2$) in which $X$ lies and $W$ identifies the bin index, relative to $V$. Thus from (\ref{eqn:rate2}) we obtain
\begin{eqnarray}
\lefteqn{H(U_2|U_1)=} & \nonumber \\
& \frac{2L_2 }{N_2}\sum_{w=1}^{N_2} H(U_2|V=2, W=w) + \nonumber \\
& ~~\frac{2L_1}{N_1} \sum_{w=1}^{N_1} H(U_2|V=1, W=w).
\end{eqnarray}
The objective is to minimize $P_{e,II}$ over $N_1$ and $N_2$ subject to the constraint that $H(U_2|U_1)+H(U_1) \leq R$. We observe that the term $H(U_2|U_1)$ is weakly dependent on $N_1$ and $N_2$ (see (\ref{eqn:weak}) below). Thus we minimize $P_{e,II}$ with a constraint on $H(U_1)$. An approximate parametric solution to this optimization problem in terms of $N_2$ is given by $N_1(N_2)=\lceil \frac{\alpha_1L_2N_2}{\alpha_2L_1}\rceil$. In terms of $N_2$ we then obtain $P_e=\alpha_1/N_1(N_2)+\alpha_2/N_2$ and $R=H(U_2|U_1) +H(\bL)+ (2L_1/L)\log_2 N_1(N_2) + (2L_2/L) \log_2 N_2$. We note here that $\alpha_2L_1/\alpha_1 L_2 < 1$ for $\pi/3 \leq \theta \leq \pi/2$. Thus $N_1 \leq N_2$ and $N_1 > 1$ only if $N_2 > \alpha_2 L_1/\alpha_1 L_2$.

\subsection{Asymptotic Analysis}
\label{sec:fine}
We study the behavior of $P_{e,II}$ and $R$ as $N_2 \rightarrow \infty$. Based on the information presented, it follows immediately that 
\begin{eqnarray}
P_{e,II} & = & \frac{\alpha_2}{N_2}(1+\frac{L_1}{L_2})(1+o(1)) \nonumber \\
R & = & H(U_2|U_1)+ H(\bL)+(2L_1/L) \log_2 \left( \frac{\alpha_1 L_2}{\alpha_2 L_1}\right) + \nonumber \\
& & ~~~(2(L_1+L_2)/L) \log_2 N_2 + o(1),
\end{eqnarray}
where $\lim_{N_2 \rightarrow \infty} o(1)=0$. Since
\begin{eqnarray}
\kappa & := & \lim_{N_2 \rightarrow \infty} H(U_2|U_1) \nonumber \\
& = &  (2/L) \int_{-1/2}^{t_{-1}} H(Q(x))dx  
\label{eqn:weak}
\end{eqnarray}
it follows  that
\begin{eqnarray}
\lefteqn{\lim_{R \rightarrow \infty} P_{e,II} 2^{LR/2(L_1+L_2)}=} & \nonumber \\
& \alpha_2 \left(1+\frac{L_1}{L_2}\right)\left(\frac{\alpha_1L_2}{\alpha_2L_1}\right)^{\frac{L_1}{(L_1+L_2)}}2^{\frac{L(\kappa+H(\bL))}{2(L_1+L_2)}}.
\end{eqnarray}
The above expression is in terms of geometric parameters of $\cB(0)$. An expression in terms of probabilities associated with $\cB(0)$ is perhaps more intuitive expression for information theorists and is obtained by defining $\bP=(P_0,P_1,P_1,P_2,P_2)$, with $P_i=L_i/L$, $i=0,1,2$. In terms of the $P_i$'s we obtain 
\begin{eqnarray}
\lefteqn{\lim_{R \rightarrow \infty} P_{e,II} 2^{R/(1-P_0)}=} & \nonumber \\
& \alpha_2 \left(1+\frac{L_1}{L_2}\right)\left(\frac{\alpha_1L_2}{\alpha_2L_1}\right)^{\frac{L_1}{(L_1+L_2)}}2^{\frac{(\kappa+H(\bP))}{(1-P_0)}}.
\end{eqnarray}
Observe that $P_0=1-\rho \cos \theta$. Thus for $\rho=1$ and $\theta \in (\pi/3,\pi/2)$, the rate at which $P_{e,II}$ decays to zero depends on $\theta$ and is maximum when $\theta \rightarrow \pi/3$.

We note here that identical results are obtained using the heavier machinery of point density functions. We have chosen to present the work using a simpler approach.

\subsection{Interactive: Single Round, Reversed Steps}
\label{sec:irev}
Analysis is now presented for  $21$ order of communication. We will summarize the description of the quantizer, and present the final results, since the derivation is similar. In fact, the derivation for this case is simpler.
The support for $X_2$ is partitioned into $2N+1$ bins. With reference to Fig.~\ref{fig:detailedGeom},  a single large bin spans the interval $J_0:=(\tau_{-1},\tau_1]$. The interval $J_{-1}:=(-\rho \sin \theta/2, \tau_{-1}]$ is partitioned into $N$ intervals of equal length $\Delta$. The same holds for the interval $J_{1}=(\tau_1,\rho \sin \theta/2]$. Equal bin sizes are justified by symmetry. Observe that there is only a single step size parameter here, as opposed to the 12 case, where two step sizes were called for. With $H=\rho \sin \theta$, the vertical ($X_2$) dimension of $\cB(0)$ and with $H_1$ as in Fig.~\ref{fig:detailedGeom}, let $H_0:=(H-2H_1)$. Let $Q:=(H_0/H, H_1/H,H_1/H)$ and $Q_0=H_0/H$. 

$S_2$ sends $U_2$, the index of the bin that $X_2$ lies in, and partitions $\cB(0)$ into horizontal strips. $S_1$ then partitions each horizontal strip into at most three parts using at most two vertical cuts or thresholds, referred to as the left and right thresholds, and sends $U_1$ to $S_2$. For a given $x_2$, let $P_{-1}(x_2)$ be the probability that $X_1$ lies to the left of the left threshold ($U_1=-1$), $P_1(x_2)$ the probability that $X_1$ lies to the right of the right threshold ($U_1=1$) and $P_0(x_2)$ the probability that $X_1$ lies in between the two thresholds ($U_1=0$). Let $P(x)=(P_{-1}(x),P_0(x),P_1(x))$. With an equivalent definition of $\kappa$, namely,
\begin{eqnarray}
\kappa & := & \lim_{N \rightarrow \infty} H(U_1|U_2) \nonumber \\
& = &  (2/H) \int_{-\rho \sin \theta/2}^{\tau_{-1}} H(P(x))dx, 
\end{eqnarray}
it follows that the total number of bits sent is given by
\begin{equation}
R=H(Q)+(1-Q_0)\log_2 N +\kappa.
\end{equation}
and 
\begin{equation}
P_{e,II}=\beta/N
\end{equation}
with $\beta=(1/2)((2L_2+L_1)/L )(H_1/H)$.
Taking limits we obtain
\begin{equation}
\lim_{N\rightarrow \infty} P_{e,II}2^{R/(1-Q_0)}=\beta 2^{(H(Q)+\kappa)/(1-Q_0)}.
\end{equation}

\section{Interactive: Infinite Rounds}
\label{sec:infinitedec}
\begin{figure}[h!]
\begin{center}
		\includegraphics[height=3.8cm]{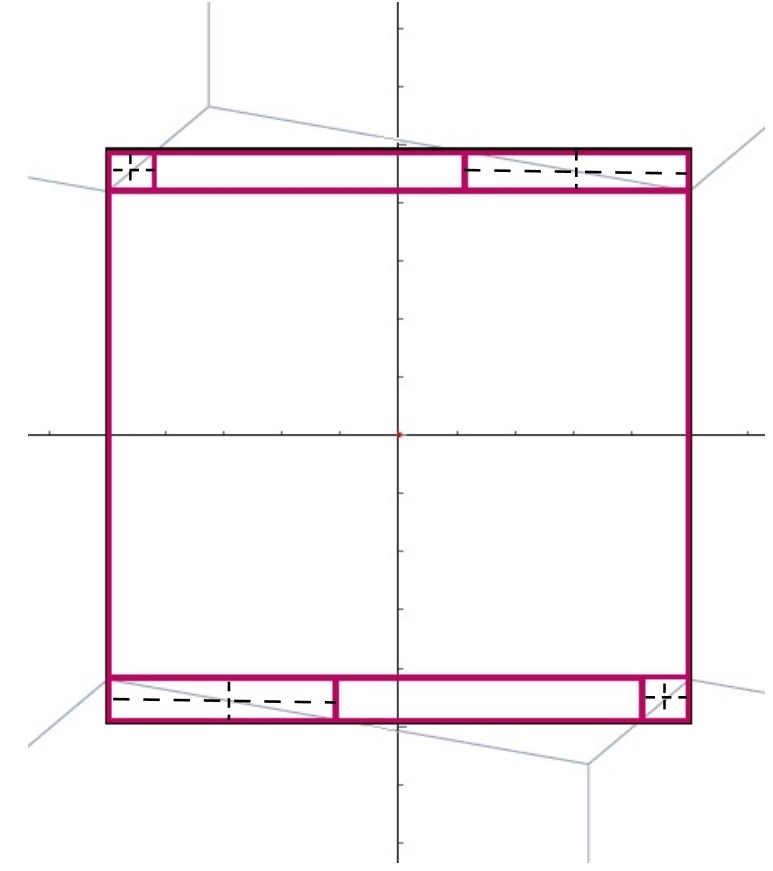}  
\caption{Red solid lines show partition after the first round of communication. Dashed lines are created in the second round of communication. }
\label{fig:pi}
\end{center}
\end{figure}

We now analyze the interactive model in which an infinite number of communication rounds are allowed. Node $S_2$ communicates first. In Round-1, Node $S_2$
partitions the support of $X_2$ into three intervals as in Sec.~\ref{sec:irev} (see Figs.~\ref{fig:pi} and \ref{fig:detailedGeom}),  $J_{-1}$,  and $J_0$, and $J_1$.  Let random variable $U_2$ be the index of the interval in which  $X_2$ lies. In Round-1, upon receiving $U_2$ and if  $U_2=1$, node $S_1$ partitions the support of $X_1$ into three intervals  $I_{-1}=(-1/2,t_{-2}]$, $I_0=(t_{-2},t_1]$ and $I_1=(t_1,1/2]$ (see Fig.~\ref{fig:detailedGeom}). If $U_2=-1$, the support of $X_2$ is partitioned into intervals $-I_1,-I_0, -I_{-1}$. If $U_2=0$, no partitioning step is taken. Random variable $U_1$ describes the interval in which $X_1$ lies. Let $Pr(U_2=i)=:Q_i$, $i=-1,0,1$. Let $P_i=Pr(U_1=i|U_2=1)$, $i=-1,0,1$. Let $Q=(Q_0,Q_1,Q_2)$ and $P=(P_0,P_1,P_2)$.

We assume that for every round, upon sending $U_i$, node $S_i$  updates $X_i$ by subtracting the lower endpoint of the interval that it lies in.

The partition of $\cB(0)$ into rectangular cells after a single, and after two rounds of communication is  shown in Fig.~\ref{fig:pi}. Define a rectangular cell to be \textit{error-free} if its interior does not contain a boundary of $\cV(0)$. 
Of the seven rectangles in the partition at the conclusion of Round-1, all but four are error-free. If $X=(X_1,X_2)$ lies in an error-free rectangle, communication halts after Round-1. Else a second round of communication occurs, during which a total of 2 bits are communicated. This process of partitioning and communication continues until each node determines that $X$ lies in an error free rectangle of the current partition.   When the algorithm halts, $P_{e,II}=0$. Let $N(X)$, $R(X)$  denote the number of rounds, and number of bits communicated, respectively, when the algorithm halts. Let $\bar{R}=E[R(X)]$ and $\bar{N}=E[N(X)]$ denote averages over $X$.
\begin{theorem}
For the interactive model with unlimited rounds of communication, a nearest plane partition can be transformed into the Voronoi partition using,  on average, a finite number of bits and rounds of communication. Specifically, 
\begin{equation}
\bar{R} =    H(Q)+(1-Q_0)H(P) + 4(1-P_0)(1-Q_0)
\label{eqn:avrateinf}
 \end{equation}
 and
 \begin{equation}
\bar{N}  =  1+2(1-P_0)(1-Q_0).
\end{equation}
\end{theorem}
\begin{proof}
We assume that an optimum entropy code is used (thus if $U_2=0$, the codeword length is $\log_2(1/Q_0)$ bits). The term $H(Q)+(1-Q_0)H(P)$ in (\ref{eqn:avrateinf}) is the cost of resolving the Round-1 partition. At the conclusion of Round-1, if $X$ belongs to a region which is not error-free, then the average number of bits transmitted is obtained by the following argument. At the conclusion of Round-1, there are two kinds of error rectangles, determined by the sign of the slope of the boundary of $\cV(0)$ in the rectangle. Note that error rectangles are designed so that the boundary of $\cV(0)$ is a diagonal of the corresponding rectangle. Let an error rectangle have length $L$ and height $H$. If the slope is positive, construct the binary expansion  $1-x_1/L=\sum_{i=1}^\infty b_i 2^{-i}$, else construct $x_1/L=\sum_{i=1}^\infty b_i 2^{-i}$. In both cases construct the binary expansion  $x_2/H=\sum_{i=1}^\infty c_i 2^{-i}$. From the independence  and uniformity of $X_1$ and $X_2$ it follows that the bits $B_i$ and $C_i$ are independent unbiased Bernoulli random variables. Further, the algorithm halts after $n$ rounds, with $2n$ total bits communicated if and only if $B_i\neq C_i$, $i<n$ and $B_n=C_n$. Thus, given $X$ in an error rectangle, $Pr(R(X)=2n)=Pr(N(X)=n)=2^{-n}$. The result follows immediately by computing the average.
\end{proof}
\begin{remark}
The communication strategy is implicit in the proof. Note that the finite value for $\bar{R}$ is because of the rapid decrease with $n$  of the probability of halting at $n$ rounds.
\end{remark}
\begin{remark}
This result has  interesting implications when viewed in the context of distributed classification problems. Suppose we have an optimum two-dimensional classifier with separating boundaries that are not axis aligned and also a suboptimal classifier with separating boundaries that are axis aligned, e.g. a $k$-$d$ tree. We expect the communication complexity of refining the approximate rectangular classifier to the optimum classifier to be finite.
\end{remark}

\section{Numerical Results and Discussion}
\label{sec:discussion}
\begin{figure}[htbp] 
   \centering
   \includegraphics[width=2.5in,height=2.0in]{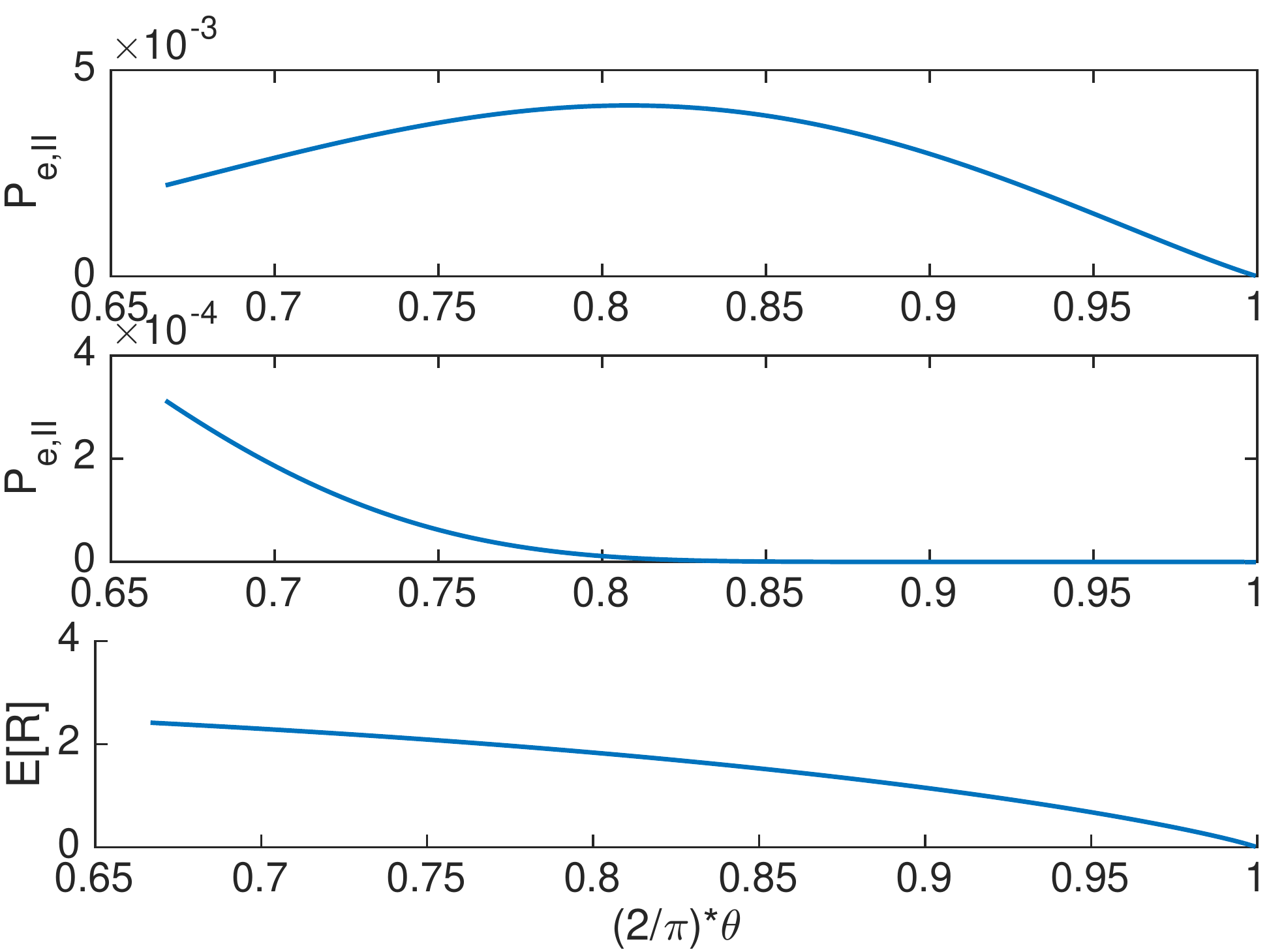} 
   \caption{Variation of $P_{e,II}$ with $\theta$ for the single-round interactive model, 12 (top), 21 (middle) with  $R=4.0$ bits. $\bar{R}=E[R]$ for the infinite-round interactive model  is shown in the bottom panel for $\rho=1$.}
   \label{fig:thetavariation}
\end{figure}
Performance results for all models are summarized in Fig.~\ref{fig:thetavariation}, for $\rho=1$ and $\pi/3 < \theta < \pi/2$.  Under the 1-round interactive model  the hexagonal lattice is not the worst case for the $12$ sequence, but is for the $21$ sequence. The large gap in performance at the same rate for the $12$ and $21$ sequences highlights the importance of selecting the sequence of order in which nodes communicate in this case. Under the infinite round interactive model, the hexagonal lattice is the worst case, with $\bar{R}=2.42$ bits.
\section{Summary and Conclusions}
\label{sec:summary}
For the nearest lattice point problem, we have considered the problem of refining an approximation to the nearest lattice point to obtain the true nearest lattice point, and have obtained the communication cost of doing so. More specifically, we have assumed that the approximate lattice point is obtained using Babai's nearest plane algorithm. The quality of the approximation has been measured by the error probability. 
An interactive communication model has been considered. 
We have shown that the rate of decay  of the error probability is  lattice dependent.  Somewhat surprisingly, the communication cost has been observed  to be finite when an infinite number of communication rounds are possible.

\end{document}